\documentclass[a4paper,11pt]{article}

\usepackage{graphicx}
\usepackage{amsmath,amsfonts}
\usepackage{amssymb}
\usepackage{color}
\usepackage[numbers,square,sort&compress]{natbib}
\usepackage{algorithm}
\usepackage{algpseudocode}
\usepackage{booktabs}
\usepackage{url}
\usepackage[subrefformat=parens]{subcaption}
\newenvironment{proof}[1][Proof]{\textbf{#1.} }{\hfill$\square$}
\newtheorem{theorem}{Theorem}
\newtheorem{lemma}{Lemma}
\newtheorem{remark}{Remark}
\newtheorem{ass}{Assumption}
\newtheorem{corollary}{Corollary}

\newtheorem{definition}{Definition}

\begin{document}

\title{Stable Inversion of Discrete-Time Linear Periodically Time-Varying Systems via Cyclic Reformulation\thanks{Corresponding author: H.~Okajima. Email: okajima@mech.kumamoto-u.ac.jp}}

\author{Hiroshi Okajima\\
{\small Faculty of Advanced Science and Technology, Kumamoto University, Japan}}

\date{}

\maketitle


\begin{abstract}
Inverse systems for discrete-time linear periodically time-varying (LPTV) plants are fundamental to feedforward control and iterative learning control of multirate and periodic systems. Building on the classical cyclic reformulation, which converts an $N$-periodic system into an equivalent LTI system at the original sampling rate, this paper derives an explicit closed-form $N$-periodic state-space realization of the inverse for an arbitrary uniform periodic relative degree $r \geq 0$ (defined through the periodic Markov parameters). The technical basis is a structure-preservation property: after absorbing a phase shift for $r \geq 1$, the LTI inverse of the cycled plant provably retains the cyclic (block-circulant/block-diagonal) structure, so that the periodic inverse matrices can be read off block-by-block, in exact correspondence with the LTI inverse. The resulting inverse system is real-valued, causal for $r = 0$ and $r$-step-delayed for $r \geq 1$, operates at the original sampling rate, and reconstructs the input exactly under matched initial conditions, with geometric error decay otherwise. Through this correspondence, its stability is shown to be equivalent to the invariant zeros of the cycled plant---the periodic invariant zeros of the plant---lying in the open unit disk, which generalizes the minimum phase condition of the LTI case and is verifiable on a single constant matrix pencil.
Numerical examples illustrate the construction, the stability characterization, and the implementation as an online periodic filter.
\end{abstract}

\section{Introduction}\label{sec:intro}

The inverse of a dynamical system plays a fundamental role in feedforward control and output tracking \citep{devasia1996tac,silverman1969tac}, iterative learning control (ILC) \citep{bristow2006csm}, and invertibility analysis \citep{hirschorn1979tac}.
For linear time-invariant (LTI) systems, stable inversion theory is classical: if all zeros lie strictly inside the unit disk, a causal stable inverse exists and can be constructed explicitly in state-space form \citep{chen1984tac,moylan1977tac}, and a bounded noncausal inverse is available otherwise via exponential dichotomy \citep{devasia1998tac,vanzundert2019ijc}, at the cost of requiring future output information (preview).

Many practical systems, however, exhibit periodically time-varying dynamics that cannot be captured by an LTI model. Linear periodically time-varying (LPTV) systems arise naturally in multirate sampled-data control \citep{chen1995sampled}, position-dependent and non-equidistantly sampled systems \citep{vanzundert2017acc,vanzundert2019ijc}, and high-precision positioning stages \citep{vanzundert2020mech}. Because the system matrices vary periodically, the transfer function framework underlying LTI inversion theory is no longer directly available.

Invertibility conditions for general linear time-varying discrete-time systems were established by \citet{kono1981ijc}. For periodic systems specifically, the inversion problem was first addressed by \citet{perdon1992}, who adapted Silverman's structure algorithm \citep{silverman1969tac} to the periodic setting (with $D_k = 0$) and obtained necessary and sufficient conditions for left/right invertibility together with a synthesis procedure for the inverse.
On the computational side, \citet{varga2004} developed the computation of generalized inverses of periodic systems via lifted pencil methods, noting explicit formulas only for the square case with invertible feedthrough as a starting point; see \citet{varga2013} for a survey of computational paradigms for linear periodic systems.
More recently, stable inversion of LPTV systems has been studied for feedforward control and ILC; a comparison of inversion-based approaches is given by \citet{vanzundert2018mech}. The Floquet-based stable inversion of \citet{vanzundert2017acc,vanzundert2019ijc} adapts the exponential dichotomy framework of \citet{devasia1998tac} to discrete-time periodic systems and handles nonminimum phase plants; that noncausal construction is exact in the infinite-horizon sense but requires complex-valued Floquet factors and processing on the doubly-infinite time axis, while for minimum phase plants their solution reduces to a direct causal inverse.
The lifting-based approaches \citep{zhu2020ajc,zhu2025tac} construct exact inverses of the lifted LTI system, including the nonminimum phase multivariable case \citep{zhu2025tac}; since they are based on lifting, the resulting inverse operates on the $N$-step block time scale and yields a finite-horizon input sequence computed offline. Signal-processing-oriented inversions of LPTV filters \citep{chauvet2003icassp} typically assume a nonzero feedthrough term and are not formulated in a state-space control framework.

This paper approaches the inversion problem through the \emph{cyclic reformulation} (also called cycling or twisting \citep{bittanti2009,kuijper1999,meyer1975}), which transforms an $N$-periodic system into an equivalent LTI system at the original sampling rate---the \emph{cycled plant}---while preserving the step-by-step input--output correspondence (Remark~\ref{rem:cyclic_vs_lifting}); the matrices of the cycled plant carry a characteristic block structure, referred to below as the cyclic structure.
Attention is restricted to discrete time, where the cyclic isomorphism remains finite-dimensional (state dimension $Nn$), and to square plants with \emph{uniform periodic relative degree} $r \geq 0$ (Assumption~\ref{ass:rd}), i.e., plants whose first nonzero periodic Markov parameter is of the same order $r$ at every time phase and is nonsingular there (Definition~\ref{def:per_markov}): this is the class in which the closed form exists, and it covers the motivating applications \citep{vanzundert2020mech,vanzundert2017acc,vanzundert2019ijc}.

Under this assumption, this paper states the closed-form $N$-periodic inverse in a unified form for every uniform $r \geq 0$, with its exact and asymptotic reconstruction regimes, and characterizes its stability through the zero structure of the plant; the latter point is where it departs from the works cited above, since \citet{perdon1992} contains no stability analysis and in \citet{vanzundert2017acc,vanzundert2019ijc} stability is checked directly on the inverse itself.
The contribution consists of two parts, corresponding to these two points:
\begin{enumerate}
\item A structure-preservation property (Theorem~\ref{thm:main}): for arbitrary uniform periodic relative degree $r \geq 0$, the LTI inverse of the cycled plant---after a phase shift of the plant output is absorbed for $r \geq 1$---has the same cyclic structure as the cycled plant itself, so that a real-valued, causal (or $r$-step-delayed) $N$-periodic state-space inverse operating at the original sampling rate can be read off from its blocks in closed form.
\item A stability characterization: the stability of this closed-form inverse is determined by the invariant zeros of the cycled plant (Theorem~\ref{thm:stability}), giving a directly verifiable periodic minimum phase condition that generalizes the classical LTI criterion.
\end{enumerate}
The resulting inverse is implementable online as a causal (or $r$-step-preview) periodic filter at the original sampling rate, and is thus directly usable for feedforward control and ILC, in which the $r$-sample output preview required for $r \geq 1$ is available because the desired output trajectory is specified in advance.
The relation of the inverse formulas to the classical constructions of \citet{perdon1992} and \citet{vanzundert2017acc}, and what the cyclic route contributes beyond them, are made explicit in Remark~\ref{rem:direct}.

This paper is organized as follows. Section~\ref{sec:methods} recalls the cyclic reformulation and the classical LTI inverse system results. Section~\ref{sec:main} presents the unified problem formulation, the closed-form inverse, and the stability characterization. Section~\ref{sec:examples} gives two numerical examples, Section~\ref{sec:discussion} discusses the relationship to existing methods, and Section~\ref{sec:conclusion} concludes the paper.

\textit{Notation:}
Let $\mathbb{R}$, $\mathbb{Z}$, $\mathbb{Z}_+$, and $\mathbb{N}$ denote the real numbers, integers, nonnegative integers, and positive integers, respectively.
For a square matrix $M$, $M^{-1}$ denotes the matrix inverse, and $\rho(M)$ denotes the spectral radius.
A matrix is called \emph{Schur stable} if all its eigenvalues satisfy $|\lambda| < 1$.
For a state-space system $(A, B, C, D)$, the \emph{invariant zeros} are the values $z_0 \in \mathbb{C}$ at which the rank of the system pencil
$\bigl[\begin{smallmatrix} z_0 I - A & -B \\ C & D \end{smallmatrix}\bigr]$
drops below its normal rank.
Throughout, all subscripts of $N$-periodic matrices (such as $A_k$, $B_k$, $C_k$, $D_k$) are interpreted modulo $N$; that is, $A_k = A_{k \bmod N}$.
For brevity, the state transition product is defined as $\Phi(j,i) := A_{j-1}A_{j-2}\cdots A_i$ for $j > i$ and $\Phi(i,i) := I$.

\section{Preliminaries}\label{sec:methods}

\subsection{Cyclic Reformulation of Periodically Time-Varying Systems}\label{sec:cyclic}

Consider a discrete-time LPTV system with period $N$ ($N \in \mathbb{N}$):
\begin{eqnarray}
x(k+1) &=& A_k x(k) + B_k u(k), \label{eq:lptv_state}\\
y(k)   &=& C_k x(k) + D_k u(k), \label{eq:lptv_output}
\end{eqnarray}
where $x \in \mathbb{R}^n$, $u \in \mathbb{R}^m$, $y \in \mathbb{R}^p$, and
$A_{k+N} = A_k$, $B_{k+N} = B_k$, $C_{k+N} = C_k$, $D_{k+N} = D_k$ for all $k \in \mathbb{Z}_+$.
This system is denoted by $(A_k, B_k, C_k, D_k)$.

The cyclic reformulation \citep{bittanti2009,cyc1,kuijper1999,meyer1975} transforms the LPTV system into an equivalent LTI system; see also \citep{grasselli1988,grasselli1991,conte1989} for the associated structural theory and \citep{multi_kf,multi2} for recent applications.
Letting $i := k \bmod N \in \{0, 1, \ldots, N-1\}$ denote the time phase, the cycled input signal $\check{u}(k) \in \mathbb{R}^{Nm}$ is defined by
\begin{equation}\label{eq:cycled_input}
\check{u}(k) = e_{i} \otimes u(k), \quad i := k \bmod N,
\end{equation}
where $e_i \in \mathbb{R}^N$ is the $(i+1)$-th standard unit vector.
The cycled state $\check{x}(k) \in \mathbb{R}^{Nn}$ and output $\check{y}(k) \in \mathbb{R}^{Np}$ are defined analogously.

The cycled system is the following LTI system:
\begin{eqnarray}
\check{x}(k+1) &=& \check{A}\check{x}(k) + \check{B}\check{u}(k), \label{eq:cycled_state}\\
\check{y}(k)   &=& \check{C}\check{x}(k) + \check{D}\check{u}(k), \label{eq:cycled_output}
\end{eqnarray}
where
\begin{equation}\label{eq:checkAB}
\check{A} = \begin{bmatrix}
0      & \cdots & \cdots & 0      & A_{N-1} \\
A_0    & 0      &        &        & 0       \\
0      & A_1    & \ddots &        & \vdots  \\
\vdots &        & \ddots & 0      & 0       \\
0      & \cdots & 0      & A_{N-2}& 0
\end{bmatrix}\!,
\end{equation}
\begin{equation}\label{eq:checkB}
\check{B} = \begin{bmatrix}
0      & \cdots & \cdots & 0      & B_{N-1} \\
B_0    & 0      &        &        & 0       \\
0      & B_1    & \ddots &        & \vdots  \\
\vdots &        & \ddots & 0      & 0       \\
0      & \cdots & 0      & B_{N-2}& 0
\end{bmatrix}\!,
\end{equation}
where $\check{A} \in \mathbb{R}^{Nn \times Nn}$, $\check{B} \in \mathbb{R}^{Nn \times Nm}$, and
\begin{equation}\label{eq:checkCD}
\check{C} = \mathrm{diag}(C_0, C_1, \ldots, C_{N-1}),
\end{equation}
\begin{equation}
\check{D} = \mathrm{diag}(D_0, D_1, \ldots, D_{N-1}),\label{eq:checkCD2}
\end{equation}
where $\check{C} \in \mathbb{R}^{Np \times Nn}$ and $\check{D} \in \mathbb{R}^{Np \times Nm}$.

The following lemma states the fundamental equivalence between the original LPTV system and the cycled LTI system, which is used throughout this paper.

\begin{lemma}[{\citealp[Chap.~6]{bittanti2009}}]\label{lem:equiv}
Let $(x(k), u(k), y(k))$ and $(\check{x}(k), \check{u}(k), \check{y}(k))$ be the trajectories of the LPTV system~(\ref{eq:lptv_state})--(\ref{eq:lptv_output}) and the cycled system~(\ref{eq:cycled_state})--(\ref{eq:cycled_output}), respectively, with the same initial condition and corresponding cycled inputs. Then:
\begin{enumerate}
\item[(i)] At each time $k$ with $i = k \bmod N$, the $i$-th block of $\check{y}(k)$ is the only nonzero block and equals $y(k)$.
\item[(ii)] Each eigenvalue $\lambda$ of the monodromy matrix $\Phi := A_{N-1}\cdots A_0$ generates $N$ eigenvalues of $\check{A}$, given by
$\lambda^{1/N}e^{j2\pi l/N}$ for $l = 0,\ldots,N-1$;
in particular, $\rho(\check{A}) = \rho(\Phi)^{1/N}$.
\item[(iii)] More generally, let $\bar{A}_k$ ($k = 0,\ldots,N-1$) be any sequence of $n \times n$ matrices, and let $\check{\bar{A}}$ denote the $Nn \times Nn$ matrix obtained by replacing $A_k$ with $\bar{A}_k$ in the $1$-shift block-circulant structure~(\ref{eq:checkAB}). Then $\rho(\check{\bar{A}}) = \rho(\bar{A}_{N-1}\cdots\bar{A}_0)^{1/N}$.
\end{enumerate}
\end{lemma}

\begin{remark}\label{rem:cyclic_vs_lifting}
Unlike the standard lifting technique \citep{chen1995sampled}, which stacks inputs and outputs over one period and evolves on the $N$-step block time scale, cyclic reformulation expands the state to dimension $Nn$ but preserves the original time scale and the step-by-step input--output correspondence of Lemma~\ref{lem:equiv}(i). This pointwise correspondence is what enables the parameter extraction step below.
\end{remark}

\subsection{Inverse Systems for LTI Plants}\label{sec:lti_inv}

This subsection recalls the classical LTI inverse system results \citep{chen1984tac,moylan1977tac,kono1981ijc,silverman1969tac}; both lemmas are standard and are stated without proof.
Consider the LTI system
\begin{eqnarray}\label{eq:lti}
G:\quad x(k+1) &=& A x(k) + B u(k), \nonumber\\
y(k) &=& Cx(k) + Du(k).
\end{eqnarray}
The system $G$ is \emph{invertible} if the transfer matrix $G(z) = C(zI-A)^{-1}B + D$ is square and has full rank for almost all $z$.
An \emph{inverse} $G^{-1}$ is a system such that $G^{-1}G = I$; for the square systems considered in this paper ($m = p$), this also implies $GG^{-1} = I$.
Throughout, the output of an inverse system is denoted by $\hat{u}$, to distinguish the reconstructed input from the plant input $u$.
For square systems, the invariant zeros of $G$ are the values $z_0 \in \mathbb{C}$ such that
\begin{equation}\label{eq:zero_poly}
\det \begin{bmatrix} z_0 I - A & -B \\ C & D \end{bmatrix} = 0.
\end{equation}

\begin{lemma}[LTI inverse, relative degree zero]\label{lem:lti_inv_reldeg0}
Suppose $D$ is square and nonsingular ($m = p$). Then $G$ is invertible and a causal inverse is given by
\begin{eqnarray}\label{eq:lti_inv}
G^{-1}:\quad \zeta(k+1) &=& (A - BD^{-1}C)\,\zeta(k) + BD^{-1}y(k), \nonumber\\
\hat{u}(k) &=& -D^{-1}C\,\zeta(k) + D^{-1}y(k).
\end{eqnarray}
The eigenvalues of $A - BD^{-1}C$ are the invariant zeros of $G$ and are the poles of $G^{-1}$.
Consequently, $G^{-1}$ is stable if and only if all invariant zeros of $G$ satisfy $|z_0| < 1$.
\end{lemma}

\begin{lemma}[LTI inverse, relative degree $r$]\label{lem:lti_inv_relr}
Suppose $D = 0$, $m = p$, and $CA^jB = 0$ for $j = 0, \ldots, r-2$, while $CA^{r-1}B$ is nonsingular. Then the system $G$ has relative degree $r$ and an $r$-step-delayed inverse is given by
\begin{eqnarray}
G^{-1}_r:\quad \zeta(k+1) &=& \bigl(A - B(CA^{r-1}B)^{-1}CA^r\bigr)\,\zeta(k) \nonumber\\
&& {}+ B(CA^{r-1}B)^{-1}y(k+r), \label{eq:lti_inv_r_state}\\
\hat{u}(k) &=& -(CA^{r-1}B)^{-1}CA^r\,\zeta(k) \nonumber\\
&& {}+ (CA^{r-1}B)^{-1}y(k+r). \label{eq:lti_inv_r_output}
\end{eqnarray}
The spectrum of $A - B(CA^{r-1}B)^{-1}CA^r$ consists of the $n - rm$ invariant zeros of $G$ (counted with multiplicity) together with exactly $rm$ additional eigenvalues at the origin; consequently, $G^{-1}_r$ is stable if and only if all invariant zeros of $G$ satisfy $|z_0| < 1$.
\end{lemma}

\section{Main Results}\label{sec:main}

\subsection{Unified Problem Formulation}\label{sec:problem}

Consider the LPTV system~(\ref{eq:lptv_state})--(\ref{eq:lptv_output}) with $m = p$ (square system).
The two cases of nonsingular feedthrough and zero feedthrough are treated in a unified manner from the outset, through the following definition.

\begin{definition}[Periodic relative degree]\label{def:per_markov}
The \emph{periodic Markov parameter of order $r$} at time phase $k$ is defined as
\begin{equation}\label{eq:per_markov}
M_k^{(0)} := D_k, \qquad
M_k^{(r)} := C_{k+r}\,\Phi(k+r, k+1)\, B_k, \quad r \geq 1,
\end{equation}
where all indices are interpreted modulo $N$ and $\Phi(k+1,k+1) = I$, so that $M_k^{(1)} = C_{k+1}B_k$.
The LPTV system has \emph{uniform periodic relative degree $r = 0$} if $M_k^{(0)} = D_k$ is nonsingular for every $k \in \{0,\ldots,N-1\}$.
It has \emph{uniform periodic relative degree $r \geq 1$} if $D_k = 0$ for all $k$, $M_k^{(j)} = 0$ for all $k$ and $j = 1, \ldots, r-1$, and $M_k^{(r)}$ is nonsingular for every $k$.
\end{definition}

\begin{ass}\label{ass:rd}
The LPTV system~(\ref{eq:lptv_state})--(\ref{eq:lptv_output}) is square ($m = p$) and has uniform periodic relative degree $r \geq 0$.
\end{ass}

\begin{remark}\label{rem:sufficiency}
The nonsingularity of $M_k^{(r)}$ at every phase $k$ in Assumption~\ref{ass:rd} is a sufficient condition for invertibility, adopted here because it keeps the inverse in explicit closed form. Invertibility holds under weaker rank-type conditions on the Markov parameters via the structure algorithm \citep{perdon1992,silverman1969tac}; the scope of such extensions is discussed in Remark~\ref{rem:general_rd}.
\end{remark}

The \emph{inverse system problem} is to find a system
\begin{eqnarray}\label{eq:inv_system}
P^{-1}_r:\quad \zeta(k+1) &=& \Gamma_k \zeta(k) + \Lambda_k y(k+r), \nonumber\\
\hat{u}(k) &=& \Omega_k \zeta(k) + \Pi_k y(k+r),
\end{eqnarray}
with $N$-periodic matrices $\Gamma_k$, $\Lambda_k$, $\Omega_k$, $\Pi_k$, such that, when $y$ is the output of $P$ driven by some input $u_{\rm ref}$, the output $\hat{u}$ of $P^{-1}_r$ reconstructs $u_{\rm ref}$ with $r$-step delay; for $r = 0$ the inverse is causal.

Two reconstruction regimes are distinguished from the outset. When the initial state of the inverse matches that of the plant, $\zeta(0) = x(0)$, the input is reconstructed \emph{exactly}: $\hat{u}(k) = u_{\rm ref}(k)$ for all $k \geq 0$. For mismatched initial conditions, the reconstruction error obeys $\epsilon(k+1) = \Gamma_k\,\epsilon(k)$ with $\epsilon(k) := x(k) - \zeta(k)$ (as established in Theorem~\ref{thm:main}), and the reconstruction is \emph{asymptotic}: $\hat{u}(k) - u_{\rm ref}(k) \to 0$ geometrically if the inverse is stable in the following sense. Both regimes are illustrated in Section~\ref{sec:examples}.

\begin{definition}[Stable inverse system]\label{def:stable_inv}
The inverse system $P^{-1}_r$ in (\ref{eq:inv_system}) is \emph{stable} if the monodromy matrix
$\Phi_{\rm inv} := \Gamma_{N-1}\cdots\Gamma_0$ is Schur stable.
\end{definition}

\subsection{Closed-Form Inverse via Cyclic Reformulation}\label{sec:inverse}

The construction proceeds in three steps:
(i)~cyclic reformulation transforms the LPTV plant into the cycled LTI system $\check{P}:(\check{A},\check{B},\check{C},\check{D})$;
(ii)~the LTI inverse $\check{P}^{-1}$ is constructed by Lemma~\ref{lem:lti_inv_reldeg0} (for $r=0$) or Lemma~\ref{lem:lti_inv_relr} (for $r \geq 1$);
(iii)~\emph{parameter extraction} recovers the $N$-periodic inverse matrices from the block structure of $\check{P}^{-1}$.
Steps (i) and (ii) are classical. The content of step (iii) is a structural invariance result: for $r \geq 1$ the cycled inverse is \emph{not} in cycled form as obtained from Lemma~\ref{lem:lti_inv_relr} (its input and feedthrough matrices carry residual shifts), and the key fact established below is that this obstruction is exactly a phase shift that can be absorbed into the output signal, after which all four matrices of the inverse provably recover the cyclic ($1$-shift block-circulant / block-diagonal) structure.

To formalize the shift structure, the following permutation matrix is used.

\begin{definition}[Cycled shift matrix]\label{def:shift}
For a positive integer $q$, the \emph{cycled shift matrix} $\check{S}_q \in \mathbb{R}^{Nq \times Nq}$ is the $1$-shift block-permutation matrix with $q \times q$ identity blocks: the $((k+1)\bmod N,\,k)$-th block of $\check{S}_q$ equals $I_q$ and all other blocks are zero.
That is, $\check{S}_q$ has the same structure as $\check{A}$ in (\ref{eq:checkAB}) with every $A_k$ replaced by $I_q$.
Its $r$-th power $\check{S}_q^r$ is the $r$-shift block-permutation matrix with $((k+r)\bmod N,\,k)$-th block equal to $I_q$.
Since $\check{S}_q^N = I_{Nq}$, one has $\check{S}_q^{-1} = \check{S}_q^{N-1} = \check{S}_q^\top$.
When the block size is clear from context, the subscript is omitted.
\end{definition}

The following lemma connects the periodic Markov parameters to the Markov parameters of the cycled LTI system; the proof relies on the shift-additivity of block-circulant products given in \citep{okajima1}.

\begin{lemma}\label{lem:markov_structure}
Under Assumption~\ref{ass:rd} with $r \geq 1$, the Markov parameters of the cycled system satisfy:
\begin{enumerate}
\item[(i)] $\check{C}\check{A}^{j-1}\check{B} = 0$ for $j = 1, \ldots, r-1$.
\item[(ii)] $\check{C}\check{A}^{r-1}\check{B}$ admits the factorization
\begin{equation}\label{eq:markov_factor}
\check{C}\check{A}^{r-1}\check{B} = \check{S}_p^{\,r}\,\mathrm{diag}\bigl(M_0^{(r)}, M_1^{(r)}, \ldots, M_{N-1}^{(r)}\bigr),
\end{equation}
where $\check{S}_p \in \mathbb{R}^{Np \times Np}$ is the cycled shift matrix with $p \times p$ identity blocks.
\end{enumerate}
In particular, $\check{C}\check{A}^{r-1}\check{B}$ is nonsingular, and the cycled system $\check{P}$ has relative degree $r$ in the LTI sense.
For $r = 0$, one has directly $\check{D} = \mathrm{diag}(M_0^{(0)},\ldots,M_{N-1}^{(0)})$, which is nonsingular.
\end{lemma}

\begin{proof}
From the definition of $\check{A}$ in (\ref{eq:checkAB}), the matrix $\check{A}^{j-1}$ is a block-circulant matrix with $(j-1)$-shift, whose $((i+j-1)\bmod N, \, i)$-th block equals $\Phi(i+j-1, i)$ (this follows by induction on $j$ from the shift-additivity of block-circulant products; see \citealp[Lemma~1]{okajima1}). The product $\check{A}^{j-1}\check{B}$ is therefore a block-circulant with $j$-shift, whose $((k+j)\bmod N, \, k)$-th block is $\Phi(k+j, k+1)B_k$. Left-multiplying by the block-diagonal $\check{C}$ preserves the shift structure and gives the $((k+j)\bmod N, \, k)$-th block equal to $C_{k+j}\Phi(k+j, k+1)B_k = M_k^{(j)}$. Under Assumption~\ref{ass:rd}, $M_k^{(j)} = 0$ for $j < r$, which proves~(i). For $j = r$, the matrix $\check{C}\check{A}^{r-1}\check{B}$ is an $r$-shift block matrix whose $((k+r)\bmod N, \, k)$-th block is $M_k^{(r)}$. This is precisely $\check{S}_p^{\,r}\,\mathrm{diag}(M_0^{(r)}, \ldots, M_{N-1}^{(r)})$, since $\check{S}_p^{\,r}$ contributes the $r$-shift permutation and the block-diagonal factor supplies the blocks $M_k^{(r)}$. Since each $M_k^{(r)}$ is nonsingular by Assumption~\ref{ass:rd}, the product is nonsingular. The statement for $r = 0$ is immediate from (\ref{eq:checkCD2}).
\end{proof}

The factorization (\ref{eq:markov_factor}), which separates the Markov parameter into a pure shift ($\check{S}_p^{\,r}$) and a block-diagonal factor containing the periodic Markov parameters, is the key to the main result below.

For later reference, the state matrix of the LTI inverse of the cycled plant $\check{P}$ (Lemma~\ref{lem:lti_inv_reldeg0} for $r = 0$; Lemma~\ref{lem:lti_inv_relr} for $r \geq 1$) is denoted by
\begin{equation}\label{eq:Ainv_def}
\check{A}_{\rm inv} :=
\begin{cases}
\check{A} - \check{B}\check{D}^{-1}\check{C}, & r = 0,\\[2pt]
\check{A} - \check{B}\bigl(\check{C}\check{A}^{r-1}\check{B}\bigr)^{-1}\check{C}\check{A}^{r}, & r \geq 1,
\end{cases}
\end{equation}
which is well defined under Assumption~\ref{ass:rd} by Lemma~\ref{lem:markov_structure} and reappears in the stability analysis of Section~\ref{sec:stability}.

\begin{theorem}\label{thm:main}
Under Assumption~\ref{ass:rd}, the $r$-step-delayed inverse system of the LPTV plant~(\ref{eq:lptv_state})--(\ref{eq:lptv_output}) is given by (\ref{eq:inv_system}) with
\begin{eqnarray}
\Gamma_k &=& A_k - B_k \bigl[M_k^{(r)}\bigr]^{-1} C_{k+r}\, \Phi(k+r, k), \label{eq:Gammak_r}\\
\Lambda_k &=& B_k \bigl[M_k^{(r)}\bigr]^{-1},           \label{eq:Lambdak_r}\\
\Omega_k  &=& -\bigl[M_k^{(r)}\bigr]^{-1} C_{k+r}\, \Phi(k+r, k),          \label{eq:Omegak_r}\\
\Pi_k     &=& \bigl[M_k^{(r)}\bigr]^{-1},               \label{eq:Pik_r}
\end{eqnarray}
where all indices are modulo $N$ and $\Phi(k+r,k) = A_{k+r-1}\cdots A_{k+1}A_k$ with $\Phi(k,k) = I$.
The inverse reconstructs $u(k)$ from $y(k+r)$: if $\zeta(0) = x(0)$, then $\hat{u}(k) = u(k)$ for all $k \geq 0$; for $\zeta(0) \neq x(0)$, the reconstruction error satisfies $\hat{u}(k) - u(k) = -\Omega_k\,\epsilon(k)$ with $\epsilon(k+1) = \Gamma_k\,\epsilon(k)$.
For $r = 0$, the inverse is causal.
\end{theorem}

\begin{proof}
\textit{Step~1 (Cyclic reformulation).}
Apply cyclic reformulation to obtain the cycled LTI system $\check{P}:(\check{A}, \check{B}, \check{C}, \check{D})$. By Lemma~\ref{lem:markov_structure}, $\check{P}$ has relative degree $r$ in the LTI sense: for $r = 0$, $\check{D}$ is nonsingular; for $r \geq 1$, $\check{D} = 0$ and $\check{C}\check{A}^{r-1}\check{B}$ is nonsingular.

\textit{Step~2 (LTI inverse construction).}
For $r = 0$, Lemma~\ref{lem:lti_inv_reldeg0} gives the causal LTI inverse $\check{P}^{-1}$ with
state matrix $\check{A}_{\rm inv}$ in (\ref{eq:Ainv_def}),
$\check{B}_{\rm inv} = \check{B}\check{D}^{-1}$,
$\check{C}_{\rm inv} = -\check{D}^{-1}\check{C}$,
$\check{D}_{\rm inv} = \check{D}^{-1}$.
For $r \geq 1$, Lemma~\ref{lem:lti_inv_relr} gives the $r$-step-delayed LTI inverse with
state matrix $\check{A}_{\rm inv}$ in (\ref{eq:Ainv_def}) and
\begin{eqnarray}
\check{B}_{\rm inv} &=& \check{B}(\check{C}\check{A}^{r-1}\check{B})^{-1}, \label{eq:Binv_r}\\
\check{C}_{\rm inv} &=& -(\check{C}\check{A}^{r-1}\check{B})^{-1}\check{C}\check{A}^r, \label{eq:Cinv_r}\\
\check{D}_{\rm inv} &=& (\check{C}\check{A}^{r-1}\check{B})^{-1}, \label{eq:Dinv_r}
\end{eqnarray}
taking $\check{y}(k+r)$ as input and producing $\check{u}(k)$ as output.

\textit{Step~3 (Parameter extraction).}
Consider first $r = 0$.
Since $\check{D}^{-1} = \mathrm{diag}(D_0^{-1},\ldots,D_{N-1}^{-1})$ is block-diagonal, the product $\check{B}\check{D}^{-1}\check{C}$ retains the $1$-shift block-circulant structure of $\check{B}$, with the $((k+1)\bmod N, \, k)$-th block equal to $B_k D_k^{-1} C_k$; hence $\check{A}_{\rm inv}$ is $1$-shift block-circulant with the corresponding block $A_k - B_k D_k^{-1} C_k$. By the same reasoning, $\check{B}_{\rm inv}$ is $1$-shift block-circulant with blocks $B_k D_k^{-1}$, while $\check{C}_{\rm inv}$ and $\check{D}_{\rm inv}$ are block-diagonal with blocks $-D_k^{-1}C_k$ and $D_k^{-1}$, respectively. Thus $\check{P}^{-1}$ possesses the cycled LTI structure, and (\ref{eq:Gammak_r})--(\ref{eq:Pik_r}) with $M_k^{(0)} = D_k$ and $\Phi(k,k) = I$ are read off block-by-block.

Consider now $r \geq 1$.
From the cycled signal definition~(\ref{eq:cycled_input}), the cycled output satisfies $\check{y}(k+r) = e_{i'} \otimes y(k+r)$ where $i' = (k+r) \bmod N$. Introducing the \emph{shifted cycled output}
\begin{equation}\label{eq:shifted_cycled_output}
\check{y}_r(k) := \check{S}_p^{-r}\,\check{y}(k+r),
\end{equation}
one verifies that $\check{y}_r(k) = e_i \otimes y(k+r)$ with $i = k \bmod N$, so $\check{y}_r(k)$ is a standard cycled signal indexed by the current time phase $k$. Substituting $\check{y}(k+r) = \check{S}_p^{\,r}\,\check{y}_r(k)$ into the cycled inverse system yields
\begin{eqnarray}
\check{\zeta}(k+1) &=& \check{A}_{\rm inv}\,\check{\zeta}(k) + \check{B}_{\rm inv}\check{S}_p^{\,r}\,\check{y}_r(k), \label{eq:cycled_inv_shifted_state}\\
\check{u}(k)   &=& \check{C}_{\rm inv}\,\check{\zeta}(k) + \check{D}_{\rm inv}\check{S}_p^{\,r}\,\check{y}_r(k). \label{eq:cycled_inv_shifted_output}
\end{eqnarray}
It is now shown that all four system matrices in (\ref{eq:cycled_inv_shifted_state})--(\ref{eq:cycled_inv_shifted_output}) have the cycled LTI structure. The block-shift bookkeeping below extends to rectangular blocks: for matrices partitioned into $N \times N$ conformable blocks, the block-shift index of a product equals the sum of the factors' block-shift indices, irrespective of the individual block sizes, a block-diagonal matrix having shift $0$.
Using the factorization $\check{C}\check{A}^{r-1}\check{B} = \check{S}_p^{\,r}\,\check{M}$ from (\ref{eq:markov_factor}), where $\check{M} := \mathrm{diag}(M_0^{(r)}, \ldots, M_{N-1}^{(r)})$, its inverse is $(\check{C}\check{A}^{r-1}\check{B})^{-1} = \check{M}^{-1}\check{S}_p^{-r}$.

\emph{Feedthrough matrix}: $\check{D}_{\rm inv}\check{S}_p^{\,r} = \check{M}^{-1}\check{S}_p^{-r}\,\check{S}_p^{\,r} = \check{M}^{-1} = \mathrm{diag}([M_0^{(r)}]^{-1}, \ldots, [M_{N-1}^{(r)}]^{-1})$, which is block-diagonal.

\emph{Input matrix}: $\check{B}_{\rm inv}\check{S}_p^{\,r} = \check{B}\,\check{M}^{-1}\check{S}_p^{-r}\,\check{S}_p^{\,r} = \check{B}\,\check{M}^{-1}$. Since $\check{B}$ is $1$-shift block-circulant and $\check{M}^{-1}$ is block-diagonal, the product is $1$-shift block-circulant with $((k+1)\bmod N,\,k)$-th block $B_k[M_k^{(r)}]^{-1}$.

\emph{Output matrix}: $\check{C}_{\rm inv} = -\check{M}^{-1}\check{S}_p^{-r}\,\check{C}\check{A}^r$. Since $\check{C}$ is block-diagonal (i.e., $0$-shift) and $\check{A}^r$ is $r$-shift, the product $\check{S}_p^{-r}\check{C}\check{A}^r$ has shift $(-r) + 0 + r = 0$, i.e., it is block-diagonal; its $k$-th diagonal block is $C_{k+r}\Phi(k+r,k)$. Hence $\check{C}_{\rm inv}$ is block-diagonal with blocks $-[M_k^{(r)}]^{-1}C_{k+r}\Phi(k+r,k)$.

\emph{State matrix}: $\check{A}_{\rm inv} = \check{A} - \check{B}\,\check{M}^{-1}\check{S}_p^{-r}\,\check{C}\check{A}^r$. Since $\check{S}_p^{-r}\check{C}\check{A}^r$ is block-diagonal and $\check{B}\,\check{M}^{-1}$ is $1$-shift block-circulant, their product is $1$-shift block-circulant with $((k+1)\bmod N,\,k)$-th block $B_k[M_k^{(r)}]^{-1}C_{k+r}\Phi(k+r,k)$. Therefore $\check{A}_{\rm inv}$ retains the $1$-shift block-circulant structure.

In summary, with $\check{y}_r(k)$ as input, the system (\ref{eq:cycled_inv_shifted_state})--(\ref{eq:cycled_inv_shifted_output}) has the standard cycled LTI structure, and $\Gamma_k$, $\Lambda_k$, $\Omega_k$, $\Pi_k$ in (\ref{eq:Gammak_r})--(\ref{eq:Pik_r}) are literally the corresponding blocks of $\check{A}_{\rm inv}$, $\check{B}_{\rm inv}\check{S}_p^{\,r}$, $\check{C}_{\rm inv}$, and $\check{D}_{\rm inv}\check{S}_p^{\,r}$, respectively. This completes the derivation of the closed-form inverse matrices.

\textit{Verification.}
For any $r \geq 0$, the relative degree assumption yields the per-step identity
\begin{equation}\label{eq:io_identity}
y(k+r) = C_{k+r}\,\Phi(k+r,k)\,x(k) + M_k^{(r)} u(k),
\end{equation}
which for $r = 0$ is the output equation itself, and for $r \geq 1$ follows from $x(k+r) = \Phi(k+r,k)x(k) + \sum_{j=0}^{r-1}\Phi(k+r,k+j+1)B_{k+j}u(k+j)$ together with $C_{k+r}\Phi(k+r,k+j+1)B_{k+j} = M_{k+j}^{(r-j)} = 0$ for $j = 1,\ldots,r-1$.
Driving (\ref{eq:inv_system}) by this $y(k+r)$ and using (\ref{eq:Gammak_r})--(\ref{eq:Pik_r}) gives $\hat{u}(k) = u(k) - \Omega_k\,\epsilon(k)$ and $\epsilon(k+1) = \Gamma_k\,\epsilon(k)$ for $\epsilon(k) = x(k) - \zeta(k)$, which proves the reconstruction claims.
\end{proof}

\begin{remark}\label{rem:direct}
The matrices read off in Step~3 admit a direct interpretation. By identity~(\ref{eq:io_identity}), the $r$-step-advanced plant is itself an $N$-periodic system with realization $(A_k,\,B_k,\,C_{k+r}\Phi(k+r,k),\,M_k^{(r)})$ and nonsingular feedthrough, and (\ref{eq:Gammak_r})--(\ref{eq:Pik_r}) are exactly the classical inverse of this realization, obtained from the inversion formula for systems with invertible feedthrough (stated for periodic systems in \citealp[eq.~(12)]{varga2004} and for time-varying systems in \citealp[Lemma~9]{vanzundert2017acc}); the Verification step above is this computation. Likewise, under Assumption~\ref{ass:rd} the Periodic Structure Algorithm of \citet{perdon1992} terminates after $r$ steps with $D_r(k) = M_k^{(r)}$ and $C_r(k) = C_{k+r}\Phi(k+r,k)$, and the left/right inverse of \citet[eqs.~(9a)--(9b)]{perdon1992} then coincides with (\ref{eq:Gammak_r})--(\ref{eq:Pik_r}). What Theorem~\ref{thm:main} adds is the exact block-by-block identification of this periodic inverse with the LTI inverse of the cycled plant. This correspondence, rather than the formulas as such, is the contribution of the cyclic reformulation: it underlies the stability characterization of Theorem~\ref{thm:stability} and provides a channel for transferring further LTI results to periodic plants at the original sampling rate.
\end{remark}

The relative degree zero case is recorded separately as a consistency check of the unified framework and to make explicit its relationship to prior work.

\begin{corollary}\label{cor:rd0}
Suppose $m = p$ and $D_k$ is nonsingular for every $k$ (uniform periodic relative degree $r = 0$). Then a causal inverse of the LPTV plant is given by (\ref{eq:inv_system}) with $r = 0$ and
\begin{eqnarray}
\Gamma_k &=& A_k - B_k D_k^{-1} C_k,\quad
\Lambda_k = B_k D_k^{-1}, \label{eq:rd0_formulas}\\
\Omega_k &=& -D_k^{-1} C_k,\quad
\Pi_k = D_k^{-1}, \nonumber
\end{eqnarray}
and this inverse satisfies $P^{-1}P = PP^{-1} = I$ as input--output operators (with matched initial conditions), the latter following by direct substitution.
\end{corollary}

\begin{proof}
The formulas follow from Theorem~\ref{thm:main} with $r = 0$, $M_k^{(0)} = D_k$, and $\Phi(k,k) = I$, and the per-step formulas themselves coincide with those appearing as a starting point in the generalized-inverse framework of \citet{varga2004}. The identities $P^{-1}P = I$ and $PP^{-1} = I$ follow from the exact reconstruction property and direct substitution, respectively.
\end{proof}

\subsection{Stability Analysis}\label{sec:stability}

The stability of the inverse system depends on whether the error dynamics $\epsilon(k+1) = \Gamma_k \epsilon(k)$ are asymptotically stable, i.e., whether the monodromy matrix
\begin{equation}\label{eq:monodromy_inv}
\Phi_{\rm inv} := \Gamma_{N-1}\cdots\Gamma_1\Gamma_0
\end{equation}
with $\Gamma_k$ given by (\ref{eq:Gammak_r}) is Schur stable. This condition is characterized through the invariant zeros of the cycled plant.

\begin{theorem}\label{thm:stability}
Under Assumption~\ref{ass:rd} (uniform periodic relative degree $r \geq 0$), the following statements are equivalent:
\begin{enumerate}
\item[(i)] The inverse system $P^{-1}_r$ of Theorem~\ref{thm:main} is stable, i.e., $\rho(\Phi_{\rm inv}) < 1$.
\item[(ii)] All invariant zeros $z_0$ of the cycled plant $\check{P}$ satisfy $|z_0| < 1$.
\item[(iii)] The matrix $\check{A}_{\rm inv}$ in (\ref{eq:Ainv_def}) is Schur stable.
\end{enumerate}
\end{theorem}

\begin{proof}
(i)$\Leftrightarrow$(iii): as established in Step~3 of the proof of Theorem~\ref{thm:main}, $\check{A}_{\rm inv}$ has the $1$-shift block-circulant structure~(\ref{eq:checkAB}) with $\Gamma_k$ in the $((k{+}1)\bmod N,\,k)$-th block, so Lemma~\ref{lem:equiv}(iii) yields $\rho(\check{A}_{\rm inv}) = \rho(\Phi_{\rm inv})^{1/N}$.
(ii)$\Leftrightarrow$(iii): by Lemma~\ref{lem:lti_inv_reldeg0} (for $r=0$) or Lemma~\ref{lem:lti_inv_relr} (for $r \geq 1$) applied to $\check{P}$, the spectrum of $\check{A}_{\rm inv}$ consists of the invariant zeros of $\check{P}$, together with additional eigenvalues at the origin when $r \geq 1$; the latter do not affect Schur stability.
\end{proof}

\begin{definition}[Periodically minimum phase]\label{def:per_min_phase}
The LPTV system~(\ref{eq:lptv_state})--(\ref{eq:lptv_output}) is called \emph{periodically minimum phase} if all invariant zeros of the cycled plant $\check{P}$ satisfy $|z_0| < 1$.
\end{definition}

\begin{corollary}\label{cor:zeros}
The LPTV system is periodically minimum phase if and only if $\rho(\Phi_{\rm inv}) < 1$, i.e., the inverse $P^{-1}_r$ of Theorem~\ref{thm:main} is stable.
\end{corollary}

\begin{remark}[Terminology]\label{rem:terminology}
The term ``periodically minimum phase'' is used here solely as a shorthand for the zero condition in Definition~\ref{def:per_min_phase}, since the invariant zeros of $\check{P}$ are exactly the poles of the inverse. In the classical usage, minimum phase additionally requires the poles of a minimal realization (equivalently, the characteristic multipliers of the plant) to lie in the open unit disk; this is a separate property of the plant that is not needed for the stability of the inverse error dynamics.
Under Assumption~\ref{ass:rd}, the system pencil of $\check{P}$ has full normal rank $Nn + Nm$, so the invariant zeros are equivalently the values $z_0$ at which
\begin{equation}\label{eq:per_zeros}
\mathrm{rank}\begin{bmatrix} z_0 I - \check{A} & -\check{B} \\ \check{C} & \check{D} \end{bmatrix} < Nn + Nm;
\end{equation}
for square systems that are reachable and observable they coincide with the transmission zeros of the transfer matrix, where it should be noted that for periodic systems the reachability and observability subspaces may be time-varying \citep{colaneri1995}. 
The invariant zeros of $\check{P}$ correspond to the \emph{periodic invariant zeros} of the LPTV system \citep{grasselli1988,grasselli1991}; since $\check{A}_{\rm inv}$ is the cycled state matrix of the periodic inverse system, Lemma~\ref{lem:equiv}(ii) shows that the nonzero invariant zeros of $\check{P}$ are the $N$-th roots of the nonzero eigenvalues of $\Phi_{\rm inv}$, so that the zero condition of Definition~\ref{def:per_min_phase} is invariant under this change of time scale; see \citet{bittanti2009} for a unified treatment.
Theorem~\ref{thm:stability} therefore characterizes stable inversion in terms of a standard structural property of the periodic system.
\end{remark}

\begin{remark}[Non-uniform relative degree]\label{rem:general_rd}
When the periodic relative degree is not uniform (e.g., $D_k$ is nonsingular at some phases and zero at others, as in multirate sampled-data systems with structural zeros in the I/O matrices), neither case of Assumption~\ref{ass:rd} applies: the cycled feedthrough $\check{D}$ is singular despite some nonsingular diagonal blocks, and the invertibility of $\check{P}$ must be analyzed through its Smith form or the structure algorithm \citep{kono1981ijc,perdon1992}. Extending the block-wise parameter extraction to this setting is nontrivial, because the cyclic structure of the inverse may not be preserved in a simple form, and is left for future work.
\end{remark}

\section{Numerical Examples}\label{sec:examples}

Two examples illustrate the unified construction of Theorem~\ref{thm:main} for $r = 1$ and $r = 2$. Both examples share the same state and input matrices $A_k$, $B_k$ and differ only in $C_k$, demonstrating how the periodic relative degree is determined by the output matrices through the periodic Markov parameters. The first example illustrates the asymptotic regime (mismatched initial condition), and the second illustrates the exact regime (matched initial condition).

Both examples use $m = p = 1$, $n = 2$, $N = 3$, and $D_k = 0$ for all $k$, with common matrices
\begin{align}\label{eq:ex_common}
A_0 &= \begin{bmatrix} 0 & 1 \\ -0.5 & 0.8 \end{bmatrix},\quad
A_1 = \begin{bmatrix} 0.3 & 0.5 \\ -0.4 & 0.6 \end{bmatrix},\quad
A_2 = \begin{bmatrix} 0.2 & 0.4 \\ -0.3 & 0.5 \end{bmatrix}, \nonumber\\
B_0 &= \begin{bmatrix} 0 \\ 1 \end{bmatrix},\quad
B_1 = \begin{bmatrix} 0.5 \\ 1 \end{bmatrix},\quad
B_2 = \begin{bmatrix} 0.5 \\ 1 \end{bmatrix}.
\end{align}
The plant monodromy $\Phi = A_2 A_1 A_0$ has spectral radius $\rho(\Phi) \approx 0.205$, so the open-loop plant is stable.

\subsection{Periodic Relative Degree $r = 1$}\label{sec:ex_rd1}

With the output matrices
\begin{equation}\label{eq:ex_rd1_C}
C_0 = \begin{bmatrix} -0.5 & 1.3 \end{bmatrix},\;
C_1 = \begin{bmatrix} 0.1 & 0.8 \end{bmatrix},\;
C_2 = \begin{bmatrix} 0.4 & 0.8 \end{bmatrix},
\end{equation}
Assumption~\ref{ass:rd} is satisfied with $r = 1$: the periodic Markov parameters are
$M_0^{(1)} = C_1 B_0 = 0.8$, $M_1^{(1)} = C_2 B_1 = 1.0$,
$M_2^{(1)} = C_0 B_2 = 1.05$,
all nonsingular.

Applying Theorem~\ref{thm:main} with $r = 1$ (where $\Phi(k\!+\!1,k) = A_k$), the inverse system matrices are:
\begin{align}
\Gamma_0 &= A_0 - B_0 [M_0^{(1)}]^{-1} C_1 A_0
= \begin{bmatrix} 0 & 1 \\ 0 & -0.125 \end{bmatrix}, \nonumber\\
\Gamma_1 &= A_1 - B_1 [M_1^{(1)}]^{-1} C_2 A_1
= \begin{bmatrix} 0.4 & 0.16 \\ -0.2 & -0.08 \end{bmatrix}, \nonumber\\
\Gamma_2 &= A_2 - B_2 [M_2^{(1)}]^{-1} C_0 A_2
\approx \begin{bmatrix} 0.433 & 0.186 \\ 0.167 & 0.071 \end{bmatrix},
\end{align}
with $\Pi_k = [M_k^{(1)}]^{-1}$, $\Lambda_k = B_k \Pi_k$, and $\Omega_k = -\Pi_k C_{k+1} A_k$ for each phase.

\textbf{Stability:}
The monodromy matrix
$\Phi_{\rm inv} = \Gamma_2\Gamma_1\Gamma_0$
has eigenvalues $\lambda_1 = 0$ and $\lambda_2 \approx 0.050$, giving $\rho(\Phi_{\rm inv}) \approx 0.050$.
The inverse is \emph{stable}: the eigenvalues of $\check{A}_{\rm inv}$ are approximately $\{0, 0, 0, 0.368, -0.184 \pm 0.319j\}$; the three nonzero ones are the invariant zeros of $\check{P}$ (all with modulus less than $1$), and the three zero eigenvalues are the additional inverse poles arising from the relative degree (Lemma~\ref{lem:lti_inv_relr}).

\textbf{Implementation:}
The matrices $\Gamma_k$, $\Lambda_k$, $\Omega_k$, $\Pi_k$ are given in closed form by (\ref{eq:Gammak_r})--(\ref{eq:Pik_r}); they are evaluated once from the plant matrices and stored for one period. Each step of $P^{-1}_1$ then requires one $n \times n$ matrix--vector product plus lower-order terms (here $n = 2$, $N = 3$), and the filter runs online indefinitely at the original sampling rate with $1$-step output preview.
For comparison, the noncausal stable-inversion operator of \citet{vanzundert2019ijc} for nonminimum phase plants requires the complex-valued matrix root $\Phi^{1/N}$ of the monodromy matrix (denoted $\Psi^{1/N}$ therein) and noncausal processing on the doubly-infinite time axis; for a periodically minimum phase plant such as the present one, their solution reduces to a direct causal inverse \citep[Corollary~17]{vanzundert2017acc}, consistent with the filter obtained here.

\textit{Simulation (asymptotic regime):}
The plant is driven by $u_{\rm ref}(k) = \sin(0.15\pi k)$ from $x(0) = [1,\; {-1}]^\top$. The $1$-step-delayed inverse is initialized with $\zeta(0) = [0,\; 0]^\top$ (mismatched initial condition). Fig.~\ref{fig:inv_sim_rd1} shows the result: despite the initial mismatch, $\hat{u}(k)$ converges to $u_{\rm ref}(k)$ within a few steps, reflecting the small spectral radius $\rho(\Phi_{\rm inv}) \approx 0.050$. The error decays by a factor of approximately $\rho(\Phi_{\rm inv})$ per period.

\begin{figure}[t]
\centering
\includegraphics[width=0.8\linewidth]{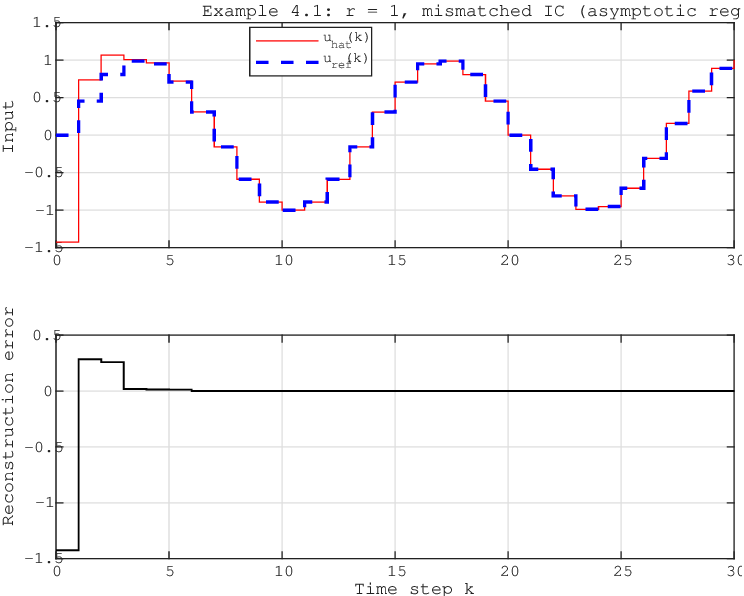}
\caption{Input reconstruction for periodic relative degree $r = 1$ (Section~\ref{sec:ex_rd1}, $N = 3$, $n = 2$, mismatched IC). Top: reference input $u_{\rm ref}(k)$ (dashed blue) and reconstructed input $\hat{u}(k)$ (solid red). Bottom: reconstruction error $\hat{u}(k) - u_{\rm ref}(k)$. The error decays geometrically due to $\rho(\Phi_{\rm inv}) \approx 0.050$.}
\label{fig:inv_sim_rd1}
\end{figure}

\subsection{Periodic Relative Degree $r = 2$}\label{sec:ex_rd2}

Using the same $A_k$, $B_k$ as in (\ref{eq:ex_common}), the output matrices are now chosen so that $C_{k+1}B_k = 0$ for all $k$, yielding $r = 2$:
\begin{equation}\label{eq:ex_rd2_C}
C_0 = \begin{bmatrix} 2 & -1 \end{bmatrix},\;
C_1 = \begin{bmatrix} 1.5 & 0 \end{bmatrix},\;
C_2 = \begin{bmatrix} 1 & -0.5 \end{bmatrix}.
\end{equation}
The constraints $C_1 B_0 = 0$, $C_2 B_1 = 0$, $C_0 B_2 = 0$ are easily verified, confirming $M_k^{(1)} = 0$ for all~$k$.
The second-order periodic Markov parameters are
$M_0^{(2)} = C_2 A_1 B_0 = 0.2$, $M_1^{(2)} = C_0 A_2 B_1 = 0.65$,
$M_2^{(2)} = C_1 A_0 B_2 = 1.5$,
all nonsingular, so Assumption~\ref{ass:rd} holds with $r = 2$.

Applying Theorem~\ref{thm:main} with $r = 2$ (where $\Phi(k\!+\!2,k) = A_{k+1}A_k$), the inverse system matrices are:
\begin{align}
\Gamma_0 &= A_0 - B_0 [M_0^{(2)}]^{-1} C_2 A_1 A_0
= \begin{bmatrix} 0 & 1 \\ 0 & -2.5 \end{bmatrix}, \nonumber\\
\Gamma_1 &= A_1 - B_1 [M_1^{(2)}]^{-1} C_0 A_2 A_1
\approx \begin{bmatrix} 0.231 & 0.092 \\ -0.538 & -0.215 \end{bmatrix}, \nonumber\\
\Gamma_2 &= A_2 - B_2 [M_2^{(2)}]^{-1} C_1 A_0 A_2
\approx \begin{bmatrix} 0.350 & 0.150 \\ 0 & 0 \end{bmatrix},
\end{align}
with $\Pi_k = [M_k^{(2)}]^{-1}$, $\Lambda_k = B_k \Pi_k$, and $\Omega_k = -\Pi_k C_{k+2} \Phi(k\!+\!2,k)$.

\textbf{Stability:}
The monodromy $\Phi_{\rm inv} = \Gamma_2\Gamma_1\Gamma_0$ is nilpotent: $\rho(\Phi_{\rm inv}) = 0$. All eigenvalues of $\check{A}_{\rm inv}$ are zero. This is structural: for $n = 2$, $N = 3$, $m = 1$, $r = 2$, the number of finite invariant zeros of $\check{P}$ is $Nn - rNm = 6 - 6 = 0$ (since Lemma~\ref{lem:lti_inv_relr} places $rNm$ eigenvalues of $\check{A}_{\rm inv}$ at the origin and the remaining $Nn - rNm$ are the invariant zeros), so the inverse is \emph{deadbeat}, with the reconstruction error vanishing within at most $n$ periods. In this example, direct computation gives $\Gamma_1\Gamma_0 = 0$ and hence $\Phi_{\rm inv} = 0$ exactly, so that with mismatched initial conditions the reconstruction error vanishes after at most one period ($N = 3$ steps).

\textit{Simulation (exact regime):}
The plant is driven by $u_{\rm ref}(k) = \sin(0.2\pi k)$ from $x(0) = [1,\; {-1}]^\top$. The $2$-step-delayed inverse (requiring $y(k+2)$ at each step) is initialized with $\zeta(0) = x(0)$ (matched initial condition). Fig.~\ref{fig:inv_sim_rd2} shows exact reconstruction from $k = 0$, with the error at machine precision, confirming the exact regime of Theorem~\ref{thm:main} for $r = 2$.

\begin{figure}[t]
\centering
\includegraphics[width=0.8\linewidth]{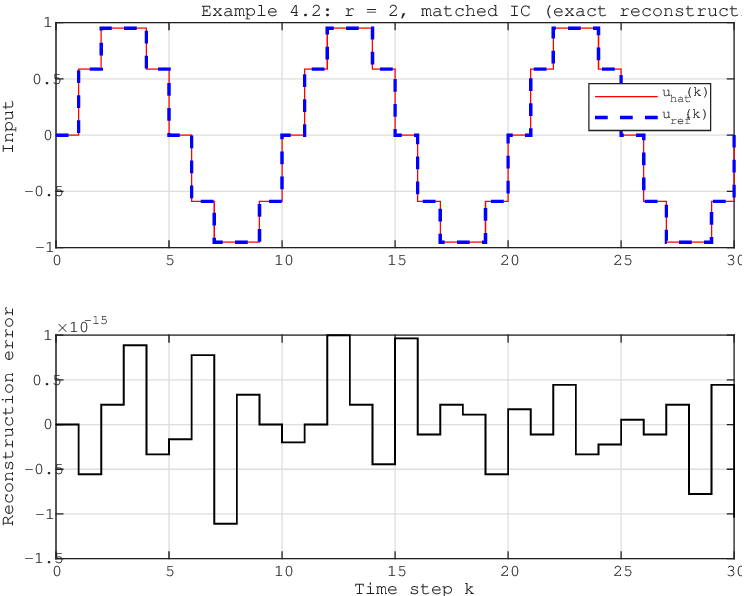}
\caption{Input reconstruction for periodic relative degree $r = 2$ (Section~\ref{sec:ex_rd2}, $N = 3$, $n = 2$, matched IC). Top: reference input $u_{\rm ref}(k)$ (dashed blue) and reconstructed input $\hat{u}(k)$ (solid red). Bottom: reconstruction error $\hat{u}(k) - u_{\rm ref}(k)$. The reconstruction is exact (error at machine precision) since $\zeta(0) = x(0)$.}
\label{fig:inv_sim_rd2}
\end{figure}

\section{Discussion}\label{sec:discussion}

At the representation level, both the lifted and the cycled representations are exact LTI equivalents of the LPTV plant, and both admit exact LTI inverses \citep{zhu2020ajc,zhu2025tac}; the essential difference is where the per-step causal structure resides (Remark~\ref{rem:cyclic_vs_lifting}). In the lifted representation, the inverse acts on the $N$-step block time scale, and within each block the per-step causal structure is encoded only implicitly (e.g., in the block-triangular structure of the lifted feedthrough), so that recovering an original-rate periodic state-space realization requires additional, nontrivial causality bookkeeping, as in the realization theory of periodic systems \citep{bittanti2009,colaneri1995}; moreover, for $D_k = 0$ the lifted feedthrough is singular, so block-shifted or delayed treatments become necessary. In the cycled representation, by contrast, Theorem~\ref{thm:main} shows that the inverse---after shift absorption---remains in cycled form, so that an original-rate periodic state-space inverse is obtained as a literal block read-off for every uniform $r \geq 0$; this structure-preservation property is precisely what yields an online causal (or $r$-step-preview) periodic filter.

Turning from representation to implementation, the proposed cyclic inverse, the lifting-based inverse \citep{zhu2020ajc,zhu2025tac}, and the Floquet-based inverse \citep{vanzundert2019ijc} differ as follows. The proposed cyclic inverse is an $N$-periodic state-space filter at the original sampling rate: it uses real arithmetic, is causal for $r = 0$ ($r$-step preview otherwise), runs online, and reconstructs the input exactly when $\zeta(0) = x(0)$, with geometric error decay at rate $\rho(\Phi_{\rm inv})$ per period otherwise; its scope is restricted to periodically minimum phase plants. The lifting-based inverse \citep{zhu2020ajc,zhu2025tac} is a finite-horizon batch feedforward computed offline on the $N$-step block time scale: it uses real arithmetic, is exact after a preparation window, and extends to nonminimum phase plants via Smith--McMillan decomposition and pseudoinversion, yielding $L_2$-bounded exact finite-time feedforward \citep{zhu2025tac}. The Floquet-based inverse \citep{vanzundert2019ijc} is a bounded noncausal operator at the original sampling rate, constructed via exponential dichotomy: it is exact in the infinite-horizon sense and likewise covers nonminimum phase plants, but requires future outputs and the complex matrix root $\Phi^{1/N}$ of the monodromy matrix, so that its factors are complex-valued even for real plants; for periodically minimum phase plants, their solution reduces to the direct causal inverse \citep[Corollary~17]{vanzundert2017acc}.
The wider scope of the lifting-based and Floquet-based routes thus reflects a different problem setting---offline finite-horizon or noncausal infinite-horizon feedforward versus online causal filtering---rather than a deficiency of either framework; the per-step cost of the proposed filter is given in Section~\ref{sec:ex_rd1}.

Turning to the algebraic inversion literature, \citet{perdon1992} gave the first invertibility and inversion results for periodic systems by adapting Silverman's structure algorithm to the periodic setting with $D_k = 0$, yielding existence conditions and a synthesis procedure for general (including non-square and non-uniform) periodic systems. Their inverse is expressed through the output matrices of a recursive Periodic Structure Algorithm and applies a time-varying bank of forward shift operators to future outputs; in general it is not given in closed form as explicit functions of the plant matrices $A_k, B_k, C_k$, although under Assumption~\ref{ass:rd} it specializes exactly to (\ref{eq:Gammak_r})--(\ref{eq:Pik_r}), as recorded in Remark~\ref{rem:direct}. They also noted (Section~5 therein) that constructing a periodic inverse via lifting remained an open problem, because no generally applicable procedure for inverting the lifting operator was then available. Beyond this specialization, the present results provide the unified formulation through the periodic Markov parameters, the two reconstruction regimes with the error dynamics $\epsilon(k+1) = \Gamma_k\epsilon(k)$ under mismatched initialization, and the stability characterization of Theorem~\ref{thm:stability}; stability of the inverse is not analyzed in \citet{perdon1992}.

On the computational side, \citet{varga2004} addressed generalized inverses of periodic systems via lifted pencil methods; the explicit formulas appear only as a two-line starting point for the square case with invertible $D_k$ (equation~(12) therein), while, for $D_k = 0$, the remaining computational aspects of the delay-inverse construction within that framework are listed as open (Section~VI-A therein). Proposition~2 therein records, as a well-known fact, that the poles of an inverse contain the transmission zeros of the plant; Theorem~\ref{thm:stability} sharpens this inclusion, for the closed-form inverse of Theorem~\ref{thm:main}, into a directly verifiable equivalence on the constant pencil~(\ref{eq:per_zeros}). \citet{varga2013} and \citet{varga2001} survey the broader computational framework for periodic systems; generalized inverses of periodic matrix pairs have also been used for model reduction of periodic descriptor systems \citep{hossain2014}, an objective different from feedforward inversion. The sequence-algebra framework of \citet{verriest2013} offers an alternative algebraic setting for periodic realizations.

Two further connections are worth noting. Theoretically, the inversion problem can be viewed as a periodic model matching problem with the identity target \citep{colaneri1997tac,tange2004acc}; the common element is the role of the periodic invariant zeros, which constrain the achievable target models there and determine the stability of the inverse here. Practically, the inverse $P^{-1}_r$ can be used directly as a feedforward controller or as the inversion model in ILC \citep{bristow2006csm,vanzundert2018mech}.

When the cycled plant has invariant zeros on or outside the unit circle, the inverse of Theorem~\ref{thm:main} fails to be stable and no stable causal realization is feasible. A natural extension, in the absence of zeros on the unit circle, is to apply the stable/unstable dichotomy \citep{devasia1998tac} to the cycled LTI plant and recover the periodic matrices of the resulting bounded noncausal inverse through parameter extraction; this is left for future work.

\section{Conclusion}\label{sec:conclusion}

This paper has derived an explicit closed-form $N$-periodic state-space realization of the inverse of a discrete-time LPTV system with uniform periodic relative degree $r \geq 0$, based on the preservation of the cyclic structure under LTI inversion of the cycled plant (after shift absorption for $r \geq 1$). The resulting inverse uses only real-valued matrices, operates at the original sampling rate, and is causal for $r = 0$ and requires $r$-step output preview otherwise; it reconstructs the input exactly under matched initial conditions, with geometric error decay otherwise, and is stable if and only if all invariant zeros of the cycled plant lie in the open unit disk. These properties make it directly applicable to feedforward control and iterative learning control of multirate and periodic systems.

Beyond the noncausal extension to nonminimum phase plants outlined in Section~\ref{sec:discussion}, future work includes systems with non-uniform relative degree across time phases (as arise in multirate I/O patterns), inversion under the weaker rank conditions of the structure algorithm \citep{perdon1992}, and integration with feedback controller design for periodic systems \citep{bittanti2009}.

\section*{Declaration of Generative AI and AI-assisted technologies in the manuscript preparation process}
During the preparation of this work the author used Claude (Anthropic) in order to refine English expressions and improve readability throughout the manuscript. After using this tool, the author reviewed and edited the content as needed and takes full responsibility for the content of the publication.


\section*{Funding}

This research did not receive any specific grant from funding agencies in the public, commercial, or not-for-profit sectors.


\end{document}